\documentclass{article}
\usepackage{arxiv}
\usepackage[OT2, T1]{fontenc}
\usepackage[russian, english]{babel}
\usepackage{caption}
\usepackage{lscape} 
\usepackage{amsmath}
\usepackage{geometry}
\usepackage{amsthm}
\usepackage{adjustbox}
\usepackage{amsfonts}
\usepackage{mathtools}
\usepackage{verbatim}
\usepackage{algpseudocode,algorithm,algorithmicx}
\usepackage{mathtools}
\usepackage{comment}
\usepackage{verbatim, times, booktabs, dcolumn, setspace, fullpage}

\newtheorem{theorem}{Theorem}

\newtheorem{definition}{Definition}

\usepackage{graphicx}
\usepackage{wrapfig}
\usepackage{caption}
\usepackage{mathtools}
\usepackage{float}
\DeclarePairedDelimiter{\norm}{\lVert}{\rVert}

\usepackage[english,noautotitles-r]{SASnRdisplay} 
\usepackage{url}
\usepackage{subfigure}
\lstdefinestyle{r-output}{
style = r-style,
style = r-output-user,
}
\DeclareMathOperator\etr{etr}
\usepackage{url}            
\usepackage{booktabs}       
\usepackage{amsfonts}       
\usepackage{nicefrac}       
\usepackage{microtype}      
\usepackage{graphicx}
\usepackage{doi}

\title{ A novel two-sample test within the space of symmetric positive definite matrix distributions and its application in finance

}

\author{ \href{https://orcid.org/0000-0002-1964-7539}{\includegraphics[scale=0.06]{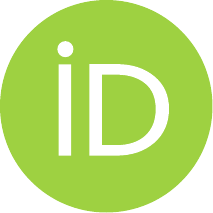}\hspace{1mm} Žikica Lukić} \\
	PhD student at the Faculty of Mathematics\\
	University of Belgrade\\
	Belgrade, 11000, Serbia \\
	\texttt{zikicamaster@gmail.com} \\
	\And
	\href{https://orcid.org/0000-0001-8243-9794}{\includegraphics[scale=0.06]{orcid.pdf}\hspace{1mm}Bojana Milošević} \\
    Faculty of Mathematics\\
	University of Belgrade\\
	Belgrade, 11000, Serbia \\
	\texttt{bojana.milosevic@matf.bg.ac.rs} \\
}

\begin{document}

\maketitle

\begin{abstract}
This paper introduces a novel two-sample test for a broad class of orthogonally equivalent positive definite symmetric matrix distributions. Our test is the first of its kind and we derive its asymptotic distribution. To estimate the test power, we use a warp-speed bootstrap method and consider the most common matrix distributions. We provide several real data examples, including the data for main cryptocurrencies and stock data of major US companies. The real data examples demonstrate the applicability of our test in the context closely related to algorithmic trading. The popularity of matrix distributions in many applications and the need for such a test in the literature are reconciled by our findings.
\end{abstract}
\section{Introduction}\label{sec::headings}

Modern computational methods have given rise to the popularity of matrix distributions in contemporary statistics and statistical learning. The distributions and their properties have recently been studied in the context of cluster analysis \cite{Gallaugher2018Finite, Tomarchio2020Two}, classification \cite{Thompson2020Classification}, and regression \cite{Ding2018Matrix}.

Positive semi-definite matrices have many applications, including medical imaging \cite{moakher2006symmetric} and finance \cite{dellaportas2012cholesky}. Positive definite matrix variate distributions have recently been studied in \cite{ouimet2022symmetric, uhler2018exact}, but the works regarding goodness-of-fit testing in this context are sparse. The work of Hadjicosta and Richards \cite{hadjicosta2020integral} is the only known goodness-of-fit (GOF) test utilising integral transform methods in this context to date. { The work of Alfelt, et al. \cite{alfelt2020goodness} considers the Bartlett decomposition to construct a goodness-of-fit test for the centralized Wishart process.}  {In addition, there has been recent interest regarding covariance matrix testing in high dimensions \cite{dornemann2023likelihood, guo2020testing, jiang2012likelihood}}.

This is in stark contrast to the current developments in the univariate case. Hankel transform GOF tests have been developed for the exponential distribution \cite{baringhaus2010empirical, baringhaus2013ks} and gamma distribution \cite{hadjicosta2020Gamma}, while closely related Laplace transforms have been used for the development of GOF tests for exponential \cite{cuparic2022new, henze2002tests, milovsevic2016new}, gamma \cite{henze2012goodness}, inverse Gaussian \cite{henze2002IG} and Rayleigh distribution \cite{meintanis2003tests}. { Baringhaus and Kolbe developed in the univariate case a two-sample goodness-of-fit test using empirical Hankel transforms \cite{baringhaus2015two}.}

We now turn our attention to the fundamental concepts that will be used in the subsequent section of this paper. Before we define the Hankel transform of a matrix random variable we introduce the following notation. We denote the matrix transpose of $X$ as $X'$, and $X>0$ indicates that $X$ is a positive definite matrix. In addition, we denote with $\etr{(X)}$ the expression $\exp(tr(X))$. We would, without explicit mention, assume that the arguments of Bessel functions are symmetric matrices, i.e. such that $X'=X$. $\mathcal{P}_+^{m\times m}$ denotes the space of positive definite $m\times m$ matrices.  $\Re$ denotes the real part of the complex number.
\begin{definition}\cite{hadjicosta2020integral}
Let $X>0$ be a symmetric random matrix with probability density function $f(X)$. For $\Re(\nu)>\frac{1}{2}(m-2)$ the Hankel transform of order $\nu$ is defined as the function
\begin{align*}
    \mathcal{H}_{X, \nu}(T)=E_X\big(\Gamma_m\big(\nu+\frac{1}{2}(m+1)\big)A_\nu(TX)\big),
\end{align*}
where $T>0$ is a symmetric matrix, $\Gamma_m$ denotes the multivariate Gamma function and $A_\nu(T)$ denotes the Bessel function of the first kind of order $\nu$.
\end{definition}
For more information regarding Bessel functions of matrix argument and closely related concepts of zonal and Laguerre polynomials, we point out to \cite{hadjicosta2019thesis, herz1955bessel, jiu2020calculation, muirhead2009aspects}.

The Hankel transform has many attractive properties. It is a continuous function of $T$, its norm is bounded from above by 1, there exists an inversion formula \cite{hadjicosta2020Gamma}, it is continuous with respect to $X$ and uniquely determines the distribution of the random variable $X$, i.e. the following theorem proved in \cite{hadjicosta2020integral} holds. 
\begin{theorem}\cite{hadjicosta2020integral}\label{teoremaHankel1}
Let $X$ and $Y$ be $m\times m$ positive definite random matrices with Hankel transforms $\mathcal{H}_{X, \nu}$ and $\mathcal{H}_{Y, \nu}$ respectively. If $\mathcal{H}_{X, \nu}=\mathcal{H}_{Y, \nu}$, then $X\overset{D}{=}Y$. 
\end{theorem}
However, the Hankel transform is not orthogonally invariant, thus being dependent on the choice of basis. An orthogonally invariant Hankel transform is  defined as follows
\begin{definition}\cite{hadjicosta2020integral}
Let $X>0$ be a random matrix with probability density function $f(X)$. For $\Re(\nu)>\frac{1}{2}(m-2)$ we define the orthogonally invariant Hankel transform of order $\nu$ as the function
\begin{align*}
    \Tilde{\mathcal{H}}_{X, \nu}(T)=E_X\big(\Gamma_m\big(\nu+\frac{1}{2}(m+1)\big)A_\nu(T, X)\big),
\end{align*}
where $T>0$, $\Gamma_m$ denotes the multivariate Gamma function and $A_\nu(T, X)$ denotes the Bessel function of the first kind of order $\nu$ with two matrix arguments. 
\end{definition}
{ We use the following notation throughout the paper:
\begin{equation*}
    J_\nu(T) = \Gamma\big(\nu + \tfrac12 (m+1)\big) A_\nu(T)
\end{equation*}
and
\begin{equation*}
J_\nu(T,X) = \Gamma\big(\nu + \tfrac12 (m+1)\big) A_\nu(T,X).    
\end{equation*}
}
We call two matrix distributions X and Y \emph{orthogonally equivalent in distribution} if there exists an orthogonal matrix $P$ such
that $X \overset{d}{=} P'Y P$. 

The following theorem provides the uniqueness of the orthogonal Hankel transform in the class of random matrices which are orthogonally equivalent in distribution.
\begin{theorem}
Let $X$ and $Y$ be $m\times m$ random symmetric positive-definite matrices and  $\Tilde{\mathcal{H}}_{X, \nu}$ and $\Tilde{\mathcal{H}}_{Y, \nu}$ their orthogonally invariant Hankel transforms respectively. Then $\Tilde{\mathcal{H}}_{X, \nu}=\Tilde{\mathcal{H}}_{Y, \nu}$ if and only if $X$ and $Y$ are orthogonally equivalent in distribution.    
\end{theorem}

\begin{proof}
From \cite[p. 260]{muirhead2009aspects}, we have that for every positive definite matrix $X$ and every symmetric matrix $Y$ the following equality holds:
\begin{equation}\label{ORTP}
   A_\nu(T, X) = \int_{\mathcal{O}(m)} A_\nu (HTH'X)dH=\int_{\mathcal{O}(m)} A_\nu (H'HTH'XH)dH=\int_{\mathcal{O}(m)} A_\nu (TH'XH)dH.
\end{equation}
The last equality emerges from the fact that matrices $X$ and $H'XH$ have the same eigenvalues if $H$ is an orthogonal matrix and $A_\nu$ is a function of eigenvalues.

Assume now that $P'YP=X$. We obtain
\begin{equation*}
   A_\nu(T, X) =\int_{\mathcal{O}(m)} A_\nu (TH'XH)dH=\int_{\mathcal{O}(m)} A_\nu (TH'P'YPH)dH=\int_{\mathcal{O}(m)} A_\nu (T(PH)'Y(PH))dH.
\end{equation*}
Note that the Haar measure is orthogonally invariant, therefore if we denote with $PH=H_1$, we obtain
\begin{align*}
       A_\nu(T, X)&=\int_{\mathcal{O}(m)} A_\nu (T(PH)'Y(PH))dH=\int_{\mathcal{O}(m)} A_\nu (TH_1'YH_1)dH_1=\\& \int_{\mathcal{O}(m)} A_\nu (HTH'Y)dH=A_\nu(T,Y).
\end{align*}
Therefore, directly from the definition of the orthogonal Hankel transform and using the fact that the Jacobian of the orthogonal transform is always equal to 1 (orthogonal transform is linear, the derivative of the linear map is equal to itself and the determinant of the orthogonal matrix equals 1), we get
\begin{align*}
    \Tilde{\mathcal{H}}_{X, \nu}(T)&=E_X\big(J_\nu(T,X)\big)=E_Y\big(J_\nu(T,P'YP)\big)=E_Y\big(J_\nu(T,Y)\big)=\Tilde{\mathcal{H}}_{Y, \nu}(T).
\end{align*}
Let's assume the equality of orthogonal Hankel transforms. The equality
\begin{align*}
    \Tilde{\mathcal{H}}_{X, \nu}(T)&=\int_{\mathcal{O}(m)}\mathcal{H}_{X, \nu}(HTH')dH
\end{align*}
holds because of the definition of Hankel transform and (\ref{ORTP}).

From the equality given above, and the definition of the orthogonally invariant Hankel transform, we have that
\begin{equation*}
    \Tilde{\mathcal{H}}_{X, \nu}(T)=E_XE_H\big(J_\nu(HTH'X)\big).
\end{equation*}
Now we use the fact that $A_\nu(\cdot)$ depends only on the eigenvalues of its arguments to obtain
\begin{equation*}
    \Tilde{\mathcal{H}}_{X, \nu}(T)=E_XE_H\big(J_\nu(THXH')\big),
\end{equation*}
but since $THXH'$ and $TZ_1$ for some $Z_1$ orthogonally equivalent to $X$ are symmetric matrices which have the same eigenvalues, and that the Jacobian of the orthogonal transform is equal to 1, we have:
\begin{equation*}
    \Tilde{\mathcal{H}}_{X, \nu}(T)=E_{Z_1}\big(J_\nu(TZ_1)\big)=\mathcal{H}_{Z_1, \nu}(T).
\end{equation*}
We obtain the similar result for $Y$, i.e. that $\Tilde{\mathcal{H}}_{Y, \nu}(T)=\mathcal{H}_{Z_2, \nu}(T)$,
and the equality of distributions $Z_1$ and $Z_2$ follows from Theorem \ref{teoremaHankel1}. Now, since $Z_1=P_1'XP_1$ and $Z_2=P_2'YP_2$, where $P_1, P_2\in \mathcal{O}(m)$, we get that $X$ is orthogonally equivalent to $Y$.
\end{proof}

Our focus will be on the orthogonally invariant Hankel transform. The rationale for this decision is based on the readily available theoretical framework presented in \cite{hadjicosta2020integral} and algorithmic solutions outlined in \cite{koev2006efficient}. Additionally, one may wish to orthogonally transform the data, such as by performing PCA or removing the dependence on the scale parameter. Orthogonal transforms are a popular method in finance as well \cite{shah2021principal, yang2015application}.


The primary goal of this paper is to extend the work presented in \cite{hadjicosta2020integral} by developing an integral two-sample test of equality of two positive definite matrix-valued distributions using the properties of orthogonal Hankel transforms. In Section \ref{sec::teststat}, the test statistic are  presented. Its asymptotic properties are  studied in Section \ref{sec::asym}.  Section \ref{sec::power} is dedicated to the  power study, while  Section \ref{sec::realdata} contains real data examples that demonstrate an application of the proposed methodology.
\section{Test statistic}\label{sec::teststat}
 Let $X=X_1, X_2, \dots, X_{n_1}$ and $Y=Y_1, Y_2, \dots, Y_{n_2}$ be two independent random samples identically distributed as $X$ and $Y$, where $X$ and $Y$ are symmetric positive definite random matrices respectively. 
Based on these samples, we  present the novel test statistic for testing the null hypothesis \begin{align*}H_0:  X \text{ and } Y \text{ are orthogonally equivalent in distribution},\end{align*}
assuming $X$ and $Y$ are symmetric positive definite random matrices. From the results in Section \ref{sec::headings}, the null hypothesis can be reformulated as:
$$H_0:  \Tilde{\mathcal{H}}_{X, \nu}(T) =  \Tilde{\mathcal{H}}_{Y, \nu}(T),\; \text{ for all } T>0. $$ 

Since the notion of orthogonal equivalence in distribution is defined in terms of the equality of corresponding orthogonal Hankel transforms, a natural way to construct a test is based on the difference of appropriate empirical counterparts, i.e.  empirical orthogonal Hankel transforms of order $\nu$  given by 

\begin{equation}\label{XEHT}
   \Tilde{\mathcal{H}}_{n_1, \nu} (T)=\Gamma_m(\nu+\frac{1}{2}(m+1))\frac{1}{n_1}\sum\limits_{j=1}^{n_1} A_\nu (T, X_j) =\frac{1}{n_1}\sum\limits_{j=1}^{n_1} J_\nu (T, X_j)
\end{equation}
and 
\begin{equation}\label{YEHT}
   \Tilde{\mathcal{H}}_{n_2, \nu} (T)=\Gamma_m(\nu+\frac{1}{2}(m+1))\frac{1}{n_2}\sum\limits_{j=1}^{n_2} A_\nu (T, Y_j)=\frac{1}{n_2}\sum\limits_{j=1}^{n_2} J_\nu (T, Y_j).
\end{equation}
For all properties of this that we will use in the sequel we refer to \cite{hadjicosta2020integral}.
It leads us to the following statistic
\begin{equation}\label{dvauzorka}
    I_{n_1, n_2,  \nu}=\int_{T>0}\Big(\Tilde{\mathcal{H}}_{{n_1}, \nu} (T)-\Tilde{\mathcal{H}}_{{n_2}, \nu} (T)\Big)^2dW(T),
\end{equation}
where  
$\Tilde{\mathcal{H}}_{n_1, \nu} (T)$ and $\Tilde{\mathcal{H}}_{n_2, \nu} (T)$  are  empirical orthogonal Hankel transform of $X$ and Y respectively, defined in  \eqref{XEHT}, and  \eqref{YEHT}, 
and $dW(T)$ is a standard Wishart measure.


For all properties of this transform that we will use in the sequel we refer to \cite{hadjicosta2020integral}.


Using the result from \cite{hadjicosta2020integral} we have 
\begin{align}\label{kvadratni_deo}
    &\int_{T>0}(\Tilde{\mathcal{H}}_{{n_1}, \nu} (T))^2dW(T)=\frac{1}{n_1^2}\sum\limits_{i=1}^{n_1}\sum\limits_{j=1}^{n_1}\etr(-X_i-X_j)J_\nu(-X_i, X_j),\\
    &\int_{T>0}(\Tilde{\mathcal{H}}_{{n_2}, \nu} (T))^2dW(T)=\frac{1}{n_2^2}\sum\limits_{i=1}^{n_2}\sum\limits_{j=1}^{n_2}\etr(-Y_i-Y_j)J_\nu(-Y_i, Y_j).
\end{align}
Following the calculations in \cite{hadjicosta2020integral}, we obtain:
\begin{align*}
   &\int_{T>0}\Tilde{\mathcal{H}}_{{n_1}, \nu} (T)\Tilde{\mathcal{H}}_{{n_2}, \nu} (T)dW(T)=\\ &\frac{\Gamma_m(\nu+\frac{1}{2}(m+1))}{n_1n_2}\sum\limits_{i=1}^{n_1}\sum\limits_{j=1}^{n_2}\int_{T>0} A_\nu(T, X_i)A_\nu(T, Y_j)(\det T)^\nu \etr(-T)dT=\\
    &\frac{1}{n_1n_2}\sum\limits_{i=1}^{n_1}\sum\limits_{j=1}^{n_2}\etr(-X_i-Y_j)J_\nu (-X_i, Y_j),
\end{align*}
and similarly
\begin{align*}
   \int_{T>0}\Tilde{\mathcal{H}}_{{n_2}, \nu} (T)\Tilde{\mathcal{H}}_{{n_1}, \nu} (T)dW(T)= \frac{1}{n_1n_2}\sum\limits_{i=1}^{n_2}\sum\limits_{j=1}^{n_1}\etr(-Y_i-X_j)J_\nu (-Y_i, X_j).
\end{align*}
Finally, we obtain that (\ref{dvauzorka}) has the following form:
\begin{align*}
    I_{n_1, n_2}=\frac{1}{n_1^2n_2^2}\sum\limits_{i=1, j=1}^{n_1}\sum\limits_{l=1, k=1}^{n_2} \Phi_\nu(X_i, X_j;Y_k, Y_l),
\end{align*}
 where $\Phi_\nu(X_i, X_j;Y_k, Y_l)$ is of the form 
 
 \begin{align*}
\Phi_\nu(X_i, X_j;Y_k, Y_l)&= \etr(-X_i-X_j)J_\nu(-X_i, X_j)+\etr(-Y_k-Y_l)J_\nu(-Y_k, Y_l)\\&-\etr(-Y_k-X_i)J_\nu (-Y_k, X_i)-\etr(-X_i-Y_k)J_\nu (-X_i, Y_k).
 \end{align*}

The evaluation of the test statistic is computationally intensive due to the complexity of evaluating the functions involved, which increases with the dimensionality of the problem.    
\section{Asymptotic properties of the novel test} \label{sec::asym}
In this section, we investigate the asymptotic behaviour of the test statistic.

It is noted in \cite{hadjicosta2020integral} that the space $L^2=L^2(W)$ of Borel measurable functions $f:\mathcal{P}^{m\times m}_+\to \mathbf{C}$ such that $\int_{X>0} |f(X)|^2dW(X)<\infty$  forms a separable Hilbert space, when equipped with the inner product
\begin{align*}
    \langle f, g \rangle =\int_{X>0} f(X)\overline{g(X)} dW(X).
\end{align*}
The norm in this space is defined as $\norm f=\sqrt{\langle f, f \rangle}$.
Assume $T$ is a symmetric positive definite $m\times m$ matrix. Let us define the stochastic process
\begin{equation*}
    \mathcal{Z}_{n_1, n_2,  \nu}(T)=\frac{1}{n_1}\sum\limits_{j=1}^{n_1} J_\nu (T, X_j)-\frac{1}{n_2}\sum\limits_{j=1}^{n_2} J_\nu (T, Y_j).
\end{equation*}
Note that although the process $\mathcal{Z}{n_1, n_2,\nu}$ depends on $\nu$, we drop the index $\nu$ for the sake of brevity. The same applies to the test statistic $I_{n_1, n_2, \nu}$.

The following inequality (\cite{hadjicosta2020integral}), 
valid for every positive definite symmetric matrix $X$ and every positive definite symmetric matrix $Y$, will be of importance in bounding the norm of $\mathcal{Z}_{n_1, n_2}$:
\begin{equation}\label{nejednakost}
    |J_\nu (X, Y)|=\Gamma_m(\nu+\frac{1}{2}(m+1))|A_\nu (X, Y)|\leq \Gamma_m(\nu+\frac{1}{2}(m+1))(\Gamma_m(\nu+\frac{1}{2}(m+1)))^{-1}=1.
\end{equation}

By applying (\ref{nejednakost}) and using the triangle inequality, we obtain that
\begin{align*}
    |\mathcal{Z}_{n_1, n_2}(T)|\leq \frac{1}{n_1}\sum\limits_{j=1}^{n_1} |J_\nu (T, X_j)|+\frac{1}{n_2}\sum\limits_{j=1}^{n_2} |J_\nu (T, Y_j)|\leq 2.
\end{align*}

Now it follows that
\begin{equation}\label{normaOGR}
    \norm{\mathcal{Z}_{n_1, n_2}(T)}^2=\int_{T>0}(\mathcal{Z}_{n_1, n_2}(T))^2dW(X)\leq \int_{T>0} 4 dW(T)=4. 
\end{equation}
Therefore, the random field $\{\mathcal{Z}_{n_1, n_2}(T), T>0\}$ is a random element of $L^2$. 
The test statistic $I_{n_1, n_2}$ can be represented as $I_{n_1, n_2} =  \norm{\mathcal{Z}_{n_1, n_2}(T)}^2$. 
Denote with $N=n_1+n_2$. We now formulate the main result of this section.
\begin{theorem}
Let $X_1, X_2, \dots, X_{n_1}$ and $Y_1, Y_2, \dots, Y_{n_2}$ be two sequences of independent orthogonally equivalent matrix random variables having the same orthogonal Hankel transform $\Tilde{\mathcal{H}}_{\nu}(T)$. Assume that $\frac{n_1}{N}\to \eta\in (0, 1)$ when $n_1, n_2\to\infty$. Then
\begin{equation*}
    \frac{n_1 n_2}{N} I_{n_1, n_2}\xrightarrow{D} \norm{\mathcal{Z}}^2,
\end{equation*}
where $\{\mathcal{Z}(T), T>0\}$ is a centered Gaussian process on $L^2$ with a covariance kernel
\begin{align*}
    \rho(S, T)=E[J_\nu(S, X)J_\nu(T, X)]-\Tilde{\mathcal{H}}_{ \nu}(S)\Tilde{\mathcal{H}}_{ \nu}(T).
\end{align*}
\end{theorem}
\begin{proof}
The proof will follow one outlined in \cite{alba2017class}.

Let $\mathcal{Z}_{n_1, X}$ and $\mathcal{Z}_{n_1, Y}$ denote the following random processes: $\mathcal{Z}_{n_1, X}(T)=\sqrt{n_1}(\Tilde{\mathcal{H}}_{n_1, \nu}(T)-\Tilde{\mathcal{H}}_{ \nu}(T))$ and $\mathcal{Z}_{n_2, Y}(T)=\sqrt{n_2}(\Tilde{\mathcal{H}}_{n_2, \nu}(T)-\Tilde{\mathcal{H}}_{ \nu}(T))$ respectively. 

Note that under $H_0$ we can write
\begin{equation*}
  \frac{n_1 n_2}{N} I_{n_1, n_2}=\norm[\bigg]{\sqrt{\frac{n_2}{N}}\mathcal{Z}_{n_1, X}-\sqrt{\frac{n_1}{N}}\mathcal{Z}_{n_2, Y}}^2.
\end{equation*}

It is worth mentioning that if constants $p$ and $q$ satisfy $p^2+q^2=1$, the process $\mathcal{Z}=p\mathcal{Z}_{1}+q\mathcal{Z}_{2}$ has the covariance structure $\rho$. Since $\sqrt{\frac{n_1}{N}}^2+\sqrt{\frac{n_2}{N}}^2=1$ and $\frac{n_1}{N}\to \eta\in (0, 1)$, it follows that $\{\sqrt{\frac{n_2}{N}}\mathcal{Z}_{n_1, X}-\sqrt{\frac{n_1}{N}}\mathcal{Z}_{n_2, Y}\}$ converges in distribution to the Gaussian process $\mathcal{Z(T)}$ with the covariance kernel $\rho$. The result of the theorem follows from the continuous mapping theorem.

Noting that for every $S>0$,  $E(\Gamma_m(\nu+\frac{1}{2}(m+1))A_\nu(S, X)-\Tilde{\mathcal{H}}_{\nu}(S))=0$, direct computation yields
\begin{align*}
    \rho(S, T)=&\text{Cov}(\Tilde{\mathcal{H}}_{n_1, \nu}(S)-\Tilde{\mathcal{H}}_{ \nu}(S), \Tilde{\mathcal{H}}_{n_1, \nu}(T)-\Tilde{\mathcal{H}}_{ \nu}(T))\\&=E[(J_\nu(S, X)-\Tilde{\mathcal{H}}_{ \nu}(S))\times(J_\nu(T, X)-\Tilde{\mathcal{H}}_{ \nu}(T))]=E[J_\nu(S, X)J_\nu(T, X)]-\Tilde{\mathcal{H}}_{ \nu}(S)\Tilde{\mathcal{H}}_{ \nu}(T).
\end{align*}
\end{proof}
It is important to emphasize that the null distribution of our test statistics depends on  the underlying distributions of $X$ and $Y$, therefore the application of certain approximation techniques is necessary to perform testing in practice. This issue will be addressed in the next section.

\section{Power study}\label{sec::power}
For practical use in dimensions $m=2$ and $m=3$, the parameter $\nu$ could be fixed to $\nu=1$, as is common practice in problems of this nature \cite{baringhaus2010empirical}. 

In this section, we present the results of the power study. The empirical powers are obtained using a warp-speed bootstrap algorithm (see, e.g., \cite{giacomini2013warp})  with $N=10000$ replications. { The method owes its name to the significant computational savings it offers in comparison to the traditional bootstrap approach.}
We provide the pseudocode for the warp-speed bootstrap algorithm below. The computation is done using MATLAB \cite{MATLAB}.
\begin{algorithm}
  \caption{Warp-speed bootstrap algorithm}
  \begin{algorithmic}[1]
    \State Sample $\textbf{x}=(x_1, \dots, x_{n_1})$ from $ F_X$ and $\mathbf{y}=y_1, \dots, y_n$ from $F_Y$;
    \State Compute $I_{n_1, n_2, \nu}(\textbf{x}, \textbf{y})$;
    \State Generate bootstrap samples $\textbf{x*}=(x^*_1, \dots, x^*_{n_1})$  and  $\textbf{y*}=y^*_1, \dots, y^*_{n_2}$ from $F_{n_1+n_2}$- sampling distribution based on the joint sample $(\textbf{x},\textbf{y})$; 
    \State Compute $I_{n_1, n_2, \nu}(\textbf{x}*,\textbf{y}*)$;
    \State Repeat steps 1-4 N times and obtain two sequences of statistics $\{I_{n_1,n_2,\nu}^{(j)}\}$ and $\{I_{n_1,n_2,\nu}^{*(j)}\}$,  $j=1,...,N$; 
    \State Reject the null hypothesis for the $j$--sample ($j=1,...,N$), if $I_{n_1,n_2,\nu}^{(j)}>c_\alpha$,  where $c_\alpha$ denotes the $(1-\alpha)\%$ quantile of the empirical distribution of the bootstrap test statistics  $(I_{n_1,n_2,\nu}^{*(j)}, \ j=1,...,N)$.
     
  \end{algorithmic}
\end{algorithm}

The level of significance is set to $\alpha=0.05$, and large values of the test statistic are considered to be significant. The algorithm developed in \cite{koev2006efficient} was implemented to evaluate the Bessel functions of two matrix arguments. { In all cases, we assume $d$ denotes the dimension of the respective matrices. When estimating sample covariance matrices, samples of dimension $d$ have been considered.} The following distributions were considered:
\begin{enumerate}
    \item Wishart distribution with the shape parameter $a$ and the scale matrix $\Sigma$, denoted by $W_d(a, \Sigma)${, with a density
    \begin{equation*}
    f_{W, a, \Sigma}(X) = \frac{1}{\Gamma_d(a)}(\det\Sigma)^a(\det X)^{a-\frac{1}{2}(d+1)}\etr(-\Sigma X);
    \end{equation*}
    
    }
    \item Inverse Wishart distribution with the shape parameter $a$ and the scale matrix $\Sigma$, denoted by $IW_d(a, \Sigma)$, { with a density
    \begin{equation*}
        f_{IW, a, \Sigma} (X)=\dfrac{(\det \Sigma)^{\frac{a}{2}}\etr(-\frac{1}{2}\Sigma X^{-1})}{2^{\frac{a d}{2}} \Gamma_d(\frac{a}{2})(\det X)^{\frac{a+d+1}{2}}};
    \end{equation*}
    }
    \item Sample covariance matrix distributions obtained from the uniform vectors $(U_1, \dots, U_d)$, where $U_i\in \mathcal{U}[0, 1]$, denoted by $CMU_d$, { with a density
    \begin{equation*}
        f_{(U_1, \dots, U_d)} ((x_1, \dots, x_d))=1, \; x_i\in [0, 1],\; 1\leq i \leq d;
    \end{equation*}
}
 \item Sample covariance matrix distributions obtained from the random vectors having the multivariate  $t$ distribution  with $a$ degrees of freedom, denoted by $CMT_d(a, \Sigma)$, { with a density
\begin{equation*}
        f_{t, a} (x)=\frac{1}{(\det(\Sigma))^\frac{1}{2}}\frac{\Gamma(\frac{a+d}{2})}{\Gamma(\frac{d}{2})(a\pi)^\frac{d}{2}}(1+\frac{x'\Sigma^{-1}x}{a})^{-\frac{a+d}{2}}.
    \end{equation*} }
\end{enumerate}

Whenever a matrix in Table \ref{pow2} or \ref{pow3} is symmetric, we leave the lower part of the table empty. It is significant that the test is able to differentiate between the different distributions with the same expectation ($W_{2}(2.5, I_2)$ versus $IW_2(4, 2.5I_2)$ and $W_3(3, I_3)$ versus $IW_3(5, 3I_3)$) with a fair degree of accuracy. 
The power of the test generally decreases with an increase in the number of dimensions. Nevertheless, the test appears to be well-calibrated, and the bootstrap approximation does not suffer from size distortions. 

Denote with $K_2$ the following covariance matrix:
$K_2=\begin{bmatrix}
\cos(0.7) & \sin(0.7) \\
 \sin (0.7) &\cos(0.7) 
\end{bmatrix}$
and denote with $K_3$ the following covariance matrix:
$K_3=\begin{bmatrix}
1 & -1 & 0.95\\
-1 & 5 & 0.01\\
0.95 & 0.01 & 7
\end{bmatrix}.$
    \begin{table}[htbp] 
\caption{Powers of the test for different sample sizes for $2\times 2$ matrices.}\label{pow2}
\begin{adjustbox}{width=\linewidth,center}

\begin{tabular}{l|lllllllllll}
\hline
$n_1=20, n_2=20$ & $W_{2}(2.5, I_2)$ & $IW_{2}(2.5, I_2)$ & $CMT_2(1, I_2)$ &$CMU_2$ & $W_{2}(2.5, 2I_2)$ & $IW_{2}(4, 2.5I_2)$ & $W_2(2.5, K_2)$ & $CMT_2(3, K_2)$ & $CMT_2(5, K_2)$ & $CMT_2(3, I_2)$ & $CMT_2(5, I_2)$ \\ \hline
$W_{2}(2.5, I_2)$ & 5 & 39 & 12 & 100 & 33 & 41 & 45  & 87 & 92 & 67 & 87 \\
$IW_{2}(2.5, I_2)$ &  & 5 & 10 & 100 & 82 & 7 & 8 & 44 & 60 & 24 & 33\\
$CMT_2(1, I_2)$ &  &  & 5 &100  & 36 & 7 & 14 & 62 & 74 & 37 & 49\\
$CMU_2$ &  &  &  & 5 & 100 & 100  & 100 & 100 & 99 & 100 & 100 \\
$W_{2}(2.5, 2I_2)$ &  &  &  &  & 4 & 96 & 89 & 97 & 99 & 91 & 95 \\
$IW_{2}(4, 2.5I_2)$ &  &  &  &  &  &  5 & 8 & 57 & 73 & 33 & 48 \\ 
$W_2(2.5, K_2)$ &  &  &  &  &  &   & 5  & 49 & 61 & 28 & 41  \\
$CMT_2(3, K_2)$ &  &  &  &  &  &   &   & 5 & 6 & 10  & 16 \\
$CMT_2(5, K_2)$ &  &  &  &  &  &   &   &  & 5 & 7 & 11 \\
$CMT_2(3, I_2)$ &  &  &  &  &  &   &   &  &  & 5 & 6 \\
$CMT_2(5, I_2)$ &  &  &  &  &  &   &   &  &  & & 5 \\

\hline
\end{tabular}
\end{adjustbox}
\medskip 
\begin{adjustbox}{width=\linewidth,center}
\begin{tabular}{l|lllllllllll}
\hline
$n_1=30, n_2=20$ & $W_{2}(2.5, I_2)$ & $IW_{2}(2.5, I_2)$ & $CMT_2(1, I_2)$ &$CMU_2$& $W_{2}(2.5, 2I_2)$ & $IW_{2}(4, 2.5I_2)$ & $W_2(2.5, K_2)$ &  $CMT_2(3, K_2)$ & $CMT_2(5, K_2)$ & $CMT_2(3, I_2)$ & $CMT_2(5, I_2)$\\ \hline
$W_{2}(2.5, I_2)$ & 5 & 49 & 19 & 100 & 38 & 46  & 57 & 93 & 97 & 79 &  89 \\
$IW_{2}(2.5, I_2)$ & 44 & 5 & 12 & 100 & 90 & 6 & 9 & 55 & 70 & 29 &  43 \\
$CMT_2(1, I_2)$ & 9 & 11 & 5  & 100  & 34 & 6 & 17 & 71 & 81 & 45 & 59 \\
$CMU_2$ & 100 & 100 & 100 & 5 & 100 & 100 & 100  & 100 & 100 & 100 &  100\\
$W_{2}(2.5, 2I_2)$ & 45 & 91 & 48 & 100  & 4 & 97 & 96 & 99 & 100 &  96 &  98\\
$IW_{2}(4, 2.5I_2)$ & 46 & 10 & 11 & 100  & 96 & 5 & 11 & 70 & 82 & 44 & 62 \\
$W_2(2.5, K_2)$ & 50 & 8 & 19 & 100 & 94 & 7  & 5 & 59 & 71 & 37 & 50 \\
$CMT_2(3, K_2)$ & 92 & 48 & 66 & 100 & 99 & 63  & 53  & 6 & 7 & 10 & 6\\
$CMT_2(5, K_2)$ & 97 & 66 & 80 & 100 & 100 & 79  & 68  & 6 & 5 &  19 &  11 \\
$CMT_2(3, I_2)$ & 75 & 22 & 42 & 100 & 95 & 35  & 31  & 12 & 21 & 5 & 7 \\
$CMT_2(5, I_2)$ & 87 & 36 & 58 & 100 & 98 & 51  & 42  & 8 & 13 & 6 & 5 \\
\hline
\end{tabular}
\end{adjustbox}
\medskip
\begin{adjustbox}{width=\linewidth,center}
\begin{tabular}{l|lllllllllll}
\hline
$n_1=50, n_2=20$ & $W_{2}(2.5, I_2)$ & $IW_{2}(2.5, I_2)$ & $CMT_2(1, I_2)$ &$CMU_2$& $W_{2}(2.5, 2I_2)$ & $IW_{2}(4, 2.5I_2)$ & $W_2(2.5, K_2)$ &  $CMT_2(3, K_2)$ & $CMT_2(5, K_2)$ & $CMT_2(3, I_2)$ & $CMT_2(5, I_2)$\\ \hline
$W_{2}(2.5, I_2)$ & 5 & 61 & 24 & 100 & 43 & 55 & 68 & 97 & 99 & 90 & 94 \\
$IW_{2}(2.5, I_2)$ & 50 & 5 & 16 & 100 & 95 & 5 & 10 & 66 & 81 & 37 & 53\\
$CMT_2(1, I_2)$ & 10 & 11 & 5  & 100  & 37 & 5 & 19 & 80 & 90 & 55 & 68 \\
$CMU_2$ & 100 & 100 & 100 & 5 & 100 & 100  & 100 & 100 & 100 & 100 & 100\\
$W_{2}(2.5, 2I_2)$ & 55 & 96 & 62 & 100  & 4 & 99 & 99 & 100 & 100 & 99 & 100 \\ 
$IW_{2}(4, 2.5I_2)$ & 53 & 11 & 20 & 100  & 98 & 5 & 14 & 80 & 90 & 54 & 70  \\ 
$W_2(2.5, K_2)$ & 56 & 8 & 26 & 100 & 98 & 7 & 5 & 70 & 82 & 46 & 60 \\
$CMT_2(3, K_2)$ & 96 & 56 & 75 & 100 & 100 & 68  & 58  & 5 & 7 & 11 & 6  \\
$CMT_2(5, K_2)$ & 99 & 73 & 87 & 100 & 100 & 85  & 77  & 7 & 5 & 22 & 12 \\
$CMT_2(3, I_2)$ & 83 & 23 & 47 & 100 & 99 & 37  & 22  & 13 & 23 & 5 & 8 \\
$CMT_2(5, I_2)$ & 93 & 39 & 64 & 100 & 100 & 55  & 50  & 8 & 14 &  6 & 5 \\
\hline
\end{tabular}
\end{adjustbox}
\medskip 
\begin{adjustbox}{width=\linewidth,center}
\begin{tabular}{l|llllllllllll}
\hline
$n_1=50, n_2=30$ & $W_{2}(2.5, I_2)$ & $IW_{2}(2.5, I_2)$ & $CMT_2(1, I_2)$ &$CMU_2$ & $W_{2}(2.5, 2I_2)$  & $IW_{2}(4, 2.5I_2)$ & $W_2(2.5, K_2)$ &  $CMT_2(3, K_2)$ & $CMT_2(5, K_2)$ & $CMT_2(3, I_2)$ & $CMT_2(5, I_2)$  \\ \hline
$W_{2}(2.5, I_2)$ & 5 & 68 & 25 & 100 & 59 & 68 & 80 & 99 & 100 & 94 & 98  \\
$IW_{2}(2.5, I_2)$ & 64 & 5 & 19 & 100 & 99 & 7 & 12 & 76 & 89 & 41 & 60 \\
$CMT_2(1, I_2)$ & 15 & 14 & 5  & 100  & 58 & 8 & 28 & 89 & 95 & 62 & 78 \\
$CMU_2$ & 100 & 100 & 100 & 5 & 100 & 100 & 100 & 100 & 100 & 100 & 100\\
$W_{2}(2.5, 2I_2)$ & 66 & 99 & 68 & 100  & 5 & 100 & 100  &100 & 100 & 100 & 100\\
$IW_{2}(4, 2.5I_2)$ & 64 & 10 & 17 & 100  & 100 & 5 & 14 & 87 & 95 & 61 & 77  \\
$W_2(2.5, K_2)$ & 73 & 10 & 36 & 100 & 100 & 8 & 5  & 78 & 90 & 54 & 69\\
$CMT_2(3, K_2)$ & 99 & 69 & 85 & 100 & 100 & 84  & 75  & 5 & 7 & 14 & 7 \\
$CMT_2(5, K_2)$ & 100 & 86 & 95 & 100 & 100 & 94  & 88  & 7 & 5 & 29 & 16 \\
$CMT_2(3, I_2)$ & 93 & 33 & 61 & 100 & 100 & 51  & 45  & 17 & 30 & 5 & 8 \\
$CMT_2(5, I_2)$ & 98 & 55 & 76 & 100 & 100 & 73  & 66  & 8 & 17 & 8 & 5 \\
\hline

\end{tabular}
\end{adjustbox}
\medskip 
\begin{adjustbox}{width=\linewidth,center}
\begin{tabular}{l|llllllllllll}
\hline
$n_1=50, n_2=50$ & $W_{2}(2.5, I_2)$ & $IW_{2}(2.5, I_2)$ & $CMT_2(1, I_2)$ &$CMU_2$ & $W_{2}(2.5, 2I_2)$ & $IW_{2}(4, 2.5I_2)$ & $W_2(2.5, K_2)$&  $CMT_2(3, K_2)$ & $CMT_2(5, K_2)$  & $CMT_2(3, I_2)$ & $CMT_2(5, I_2)$  \\ \hline
$W_{2}(2.5, I_2)$ & 5 & 79 & 26 & 100 & 77 & 78 & 89 & 100 & 100 & 98 & 100\\
$IW_{2}(2.5, I_2)$ &  & 5 & 19 & 100 & 100 & 9 & 15 & 84 & 95 & 48 & 71\\
$CMT_2(1, I_2)$ &  &  & 6 &100  & 78 & 16 & 49 & 95 & 99 & 74 & 89\\
$CMU_2$ &  &  &  & 5 & 100 & 100 & 100 & 100 & 100 & 100 & 100 \\
$W_{2}(2.5, 2I_2)$ &  &  &  &  & 4  & 100 & 100 & 100 & 100 & 100 & 100\\ 
$IW_{2}(4, 2.5I_2)$ &  &  &  &  &   & 5 & 15 & 94 & 99 & 69 & 86\\
$W_2(2.5, K_2)$ &  &  &  &  &   &  & 5 & 87 & 95 & 64 & 80\\
$CMT_2(3, K_2)$ &  &  &  &  &  &   &   & 5 & 8 & 19 & 37 \\
$CMT_2(5, K_2)$ &  &  &  &  &  &   &   &  & 5 & 9 & 20 \\
$CMT_2(3, I_2)$ &  &  &  &  &  &   &   &  &  & 5 & 9 \\
$CMT_2(5, I_2)$ &  &  &  &  &  &   &   &  &  & & 5 \\
\hline
\end{tabular}
\end{adjustbox}
\end{table}


\begin{table}[htbp]
\caption{Powers of the test for different sample sizes for $3\times 3$ matrices.}\label{pow3}
\begin{adjustbox}{width=\linewidth,center}
\begin{tabular}{l|lllllllllllll}
\hline
$n_1=20, n_2=20$ & $W_3(3, I_3)$ & $IW_3(3, I_3)$ & $CMT_3(1, I_3)$ &$CMU_3$& $W_3(3, 2I_3)$ & $IW_3(5, 3I_3)$ & $W_3(3, K_3)$ & $CMT_3(3, K_3)$ & $CMT_3(5, K_3)$  & $CMT_3(3, I_3)$ & $CMT_3(5, I_3)$\\ \hline
$W_3(3, I_3)$ & 5 & 25 & 6 & 100 & 16 & 51 & 19 & 57 & 71 & 52 & 69 \\
$IW_3(3, I_3)$ &  & 5  & 20 & 100 & 39  & 5 & 42 & 22 & 35 & 17 & 29\\
$CMT_3(1, I_3)$ &  &  & 5  & 100  & 10  & 33 & 11 & 51 & 67 & 47 & 61 \\
$CMU_3$ &  &  &  &  5 &  100 & 100 & 100 & 100 & 100 & 100 & 100\\
$W_3(3, 2I_3)$ &  &  &  &  &  6 & 80 & 8 & 63 & 80 & 63 & 78 \\
$IW_3(5, 3I_3)$ &  &  &  &  &  &  5  & 80 & 23 & 36 & 19 & 29\\
$W_3(3, K_3)$ &  &  &  &  &  &    & 6 & 66 & 80 & 63 & 78\\
$CMT_3(3, K_3)$ &  &  &  &  &  &   &   & 4 & 6 & 5 & 8  \\
$CMT_3(5, K_3)$ &  &  &  &  &  &   &   &  & 5 & 5 & 6 \\
$CMT_3(3, I_3)$ &  &  &  &  &  &   &   &  &  & 5 & 6 \\
$CMT_3(5, I_3)$ &  &  &  &  &  &   &   &  &  &  & 5 \\
\hline

\end{tabular}
\end{adjustbox}
\medskip 
\begin{adjustbox}{width=\linewidth,center}
\begin{tabular}{l|lllllllllll}
\hline
$n_1=30, n_2=20$ & $W_3(3, I_3)$ & $IW_3(3, I_3)$ & $CMT_3(1, I_3)$ &$CMU_3$& $W_3(3, 2I_3)$ & $IW_3(5, 3I_3)$ & $W_3(3, K_3)$ & $CMT_3(3, K_3)$ & $CMT_3(5, K_3)$ & $CMT_3(3, I_3)$ & $CMT_3(5, I_3)$\\ \hline
$W_3(3, I_3)$ & 4 & 36 & 7 & 100 & 14 & 65 & 17 & 70 & 85 & 66 & 80 \\
$IW_3(3, I_3)$ & 25 & 4 & 21 & 100 & 41 & 3 & 45 & 29 & 43 & 23 & 38 \\
$CMT_3(1, I_3)$ & 4 & 27 & 5  & 100  & 8 & 43 & 8 & 64 & 78 & 60 & 75  \\
$CMU_3$ & 100 & 100 & 100 & 5 & 100 & 100 & 100 & 100 & 100 & 100 & 100 \\
$W_3(3, 2I_3)$ & 23 & 54 & 16 & 100  & 6 & 89 & 7 & 78 & 90 & 77 & 88\\
$IW_3(5, 3I_3)$ & 56 & 7 & 44 & 100  & 84 & 5 & 88 & 33 & 48 & 26 & 43 \\
$W_3(3, K_3)$ & 30 & 57 & 17 & 100 & 10 &  90  & 6 & 79 & 92 & 79 & 89\\
$CMT_3(3, K_3)$ & 62 & 21 & 56 & 100 & 71 & 20  & 74  & 5 & 7 & 5 & 10  \\
$CMT_3(5, K_3)$ & 80 & 36 & 73 & 100 & 88 & 38  & 88  & 6 & 5 & 5 & 6 \\
$CMT_3(3, I_3)$ & 57 & 17 & 53 & 100 & 67 & 18  & 71  & 5 & 5 & 5 & 7 \\
$CMT_3(5, I_3)$ & 74 & 31 & 70 & 100 & 87 & 31  & 87  & 8 & 6 & 7 & 5 \\

\hline
\end{tabular}
\end{adjustbox}
\medskip
\begin{adjustbox}{width=\linewidth,center}
\begin{tabular}{l|lllllllllll}
\hline
$n_1=50, n_2=20$ & $W_3(3, I_3)$ & $IW_3(3, I_3)$ & $CMT_3(1, I_3)$ &$CMU_3$ & $W_3(3, 2I_3)$  & $IW_3(5, 3I_3)$ & $W_3(3, K_3)$ & $CMT_3(3, K_3)$ & $CMT_3(5, K_3)$ & $CMT_3(3, I_3)$ & $CMT_3(5, I_3)$ \\ \hline
$W_3(3, I_3)$ & 5 & 52 & 10 & 100 & 12 & 75 & 14 & 83 & 93 & 81 & 91 \\
$IW_3(3, I_3)$ & 26 & 4 & 21 & 100 & 46 & 3 & 48 & 36 & 55 & 31 & 48\\
$CMT_3(1, I_3)$ & 2 & 37 & 5  & 100  & 4 & 50  & 6 & 78 & 89 & 73 & 88 \\
$CMU_3$ & 100 & 100 & 100 & 5 & 100 & 100 & 100 & 100 & 100 & 100 & 100 \\
$W_3(3, 2I_3)$ & 36 & 72 & 23 & 100  & 5 & 96 & 4 & 90 & 97 & 90 & 96\\
$IW_3(5, 3I_3)$ & 66 & 9 & 53 & 100  & 93 & 5 & 95 & 44 & 62 & 36 & 56\\ 
$W_3(3, K_3)$ & 45 & 73 & 25 & 100 & 13 &  97  & 6 & 92 & 97 & 90 & 96 \\ 
$CMT_3(3, K_3)$ & 70 & 22 & 60 & 100 & 83 & 20  & 83  & 5 & 8 & 5 & 6  \\
$CMT_3(5, K_3)$ & 86 & 41 & 82 & 100 & 95 & 40  & 95  & 6 & 5 & 8 & 4 \\
$CMT_3(3, I_3)$ & 65 & 17 & 61 & 100 & 81 & 15  & 81  & 6 & 12 & 5 & 8 \\
$CMT_3(5, I_3)$ & 84 & 33 & 78 & 100 & 94 & 32  & 96  & 5 & 7 & 7 & 5 \\
\hline
\end{tabular}
\end{adjustbox}
\medskip 
\begin{adjustbox}{width=\linewidth,center}
\begin{tabular}{l|lllllllllll}
\hline
$n_1=50, n_2=30$ & $W_3(3, I_3)$ & $IW_3(3, I_3)$ & $CMT_3(1, I_3)$ &$CMU_3$ & $W_3(3, 2I_3)$  & $IW_3(5, 3I_3)$ & $W_3(3, K_3)$ & $CMT_3(3, K_3)$ & $CMT_3(5, K_3)$ & $CMT_3(3, I_3)$ & $CMT_3(5, I_3)$  \\ \hline
$W_3(3, I_3)$ & 4 & 55 & 8 & 100 & 21 & 84 & 29 & 90 & 97 & 87 & 97  \\
$IW_3(3, I_3)$ & 44 & 5 & 33 & 100 & 70 & 4  & 70 & 43 & 62 & 35 & 54\\
$CMT_3(1, I_3)$ & 3 & 41 & 5  & 100  & 10 & 61 & 11 & 85 & 95 & 82 & 93 \\
$CMU_3$ & 100 & 100 & 100 & 5 & 100 & 100 & 100 & 100 & 100 & 100 & 100 \\
$W_3(3, 2I_3)$ & 37 & 77 & 22 & 100  & 5 & 99 & 5 & 96 & 99 & 95 & 99\\
$IW_3(5, 3I_3)$ & 79 & 7 & 63 & 100  & 99 & 5 & 99 & 47 & 67 & 40 &  60  \\
$W_3(3, K_3)$ & 47 & 80 & 24 & 100 & 11 &  99  & 6  & 96 & 99 &  96 & 99\\
$CMT_3(3, K_3)$ & 86 & 33 & 80 & 100 & 94 & 32  & 94  & 5 & 8 & 5 &  5 \\
$CMT_3(5, K_3)$ & 96 & 55 & 93 & 100 & 99 & 60  & 99  & 7 & 5 &  11 & 6 \\
$CMT_3(3, I_3)$ & 83 & 26 & 76 & 100 & 94 & 25  & 94  & 6 & 12 & 5 & 9  \\
$CMT_3(5, I_3)$ & 95 & 48 & 91 & 100 & 99 & 51  & 99  & 6 & 7 & 7 & 6 \\
\hline
\end{tabular}
\end{adjustbox}
\medskip 
\begin{adjustbox}{width=\linewidth,center}
\begin{tabular}{l|llllllllllll}
\hline
$n_1=50, n_2=50$ & $W_3(3, I_3)$ & $IW_3(3, I_3)$ & $CMT_3(1, I_3)$ &$CMU_3$ & $W_3(3, 2I_3)$ & $IW_3(5, 3I_3)$ & $W_3(3, K_3)$ & $CMT_3(3, K_3)$ & $CMT_3(5, K_3)$ & $CMT_3(3, I_3)$ & $CMT_3(5, I_3)$\\ \hline
$W_3(3, I_3)$ & 5 & 63 & 6 & 100 & 39 & 92 & 50 & 96 & 99 & 95 & 99\\
$IW_3(3, I_3)$ &  &  6 & 48  &  100 &  86 & 5 & 88 & 47 & 73 & 39 & 64 \\
$CMT_3(1, I_3)$ &  &  &  4 & 100  & 19   & 76 & 21 & 93 & 98 & 90 & 98\\
$CMU_3$ &  &  &  & 5  & 100  & 100 & 100 & 100 & 100 & 100 & 100 \\
$W_3(3, 2I_3)$ &  &  &  &  &  5  & 100 & 9 & 99 & 100 & 99 & 100 \\ 
$IW_3(5, 3I_3)$ &  &  &  &  &   & 5 & 100 & 54 & 78 & 43 & 68 \\ 
$W_3(3, K_3)$ &  &  &  &  &  &    & 6 & 99 & 100 & 99 & 100 \\ 
$CMT_3(3, K_3)$ &  &  &  &  &  &   &   & 5 & 9 & 6 & 15 \\
$CMT_3(5, K_3)$ &  &  &  &  &  &   &   &  & 5 & 5 & 6 \\
$CMT_3(3, I_3)$ &  &  &  &  &  &   &   &  &  & 5 & 9 \\
$CMT_3(5, I_3)$ &  &  &  &  &  &   &   &  &  &  & 6 \\

\hline
\end{tabular}
\end{adjustbox}
\end{table}


\section{Real data examples}\label{sec::realdata} 
Matrix-valued distributions have been used recently to model stock market data \cite{hadjicosta2020integral, haff2011minimax}. However, despite the popularity of cryptocurrencies in the scientific community, such methods have not yet been implemented for the cryptocurrency market. The extreme volatility of cryptocurrency markets \cite{liu2019volatility} makes it essential for traders to be aware of any significant changes in the statistical properties of assets. The properties of logarithmic returns can be of interest in analysing the cryptocurrency market \cite{bibinger2019estimating, chu2015statistical, pennoni2022exploring}.
{ The correlation structure of major cryptocurrency prices was investigated in \cite{kumar2019}.}

We have looked at hourly data of two of the biggest cryptocurrencies, Bitcoin (BTC) and Ethereum (ETH). The data was downloaded from Gemini (\href{http://www.gemini.com}{http://www.gemini.com}). We have selected two periods. The first period consists of days between 01 January 2019 and 01 March 2019. In this period, the market dynamic experienced no significant changes (see Fig. \ref{fig:foobar} and \ref{fig:fooLR}). The second one consists of days between 01 April 2021 and 01 June 2021. In this period, the market experienced a positive bubble followed by a negative bubble \cite{shu20212021}.  The first period consists of $n_1=1416$ points, while the second period consists of $n_2=1464$ points. Daily observations consist of 24 points, each corresponding to one hour.

We partition the hourly close prices $X_t$ into daily periods consisting of 24 observations. We calculate hourly logarithmic returns $\log(\frac{X_{t}}{X_{t-1}})$ and calculate the $2\times 2$ unnormalized covariance matrix for each day. 

We have computed 59 covariance matrices for the first period and split it to 31 corresponding to January and 28 corresponding to February. The p-value of the statistic $K_{31, 28}$ equals 0.9756, which points out that the statistic $K$ detected no significant change in the covariance structure of the asset hourly logarithmic returns in the period January - March 2019. 

In the case of the second period, we obtained 61 $2\times 2$ matrices which we have split into 31 corresponding to April and 30 corresponding to May. We estimated the p-value using bootstrap and $N = 10000$ replications. We obtained that the p-value of the statistic $K_{31,30}$ equals 0.0003, which points out that the covariance structure of the asset hourly logarithmic returns has significantly changed from April to May 2021.

Furthermore, we have analyzed the behaviour of the covariance structure of the hourly logarithmic returns 15 days before and 15 days after the well-documented drops in Bitcoin price that coincide with prominent historical events \cite{fruehwirt2021cumulation}. Every period has $n=720$ points, and the results are presented in Table \ref{pvals}. In most cases the test is unable to detect the change in the covariance structure. Furthermore, we conducted an analysis of 1-minute BTC \cite{BTCdata} and ETH \cite{ETHdata} data. We selected two-day periods and computed the minute logarithmic returns. We computed covariance matrices for every hour, resulting in a total of $n=48$ covariance matrices, with 24 for the first day and 24 for the second day. The results presented in Table \ref{pvals1min} indicate that our test has detected statistically significant difference in the distribution of covariance matrices both on the day and a day after the occurrence of the event, while it hasn't detected statistically significant difference a day prior to the occurrence of the event. This could be of significant practical importance in implementing stop losses \cite{kaminski2014stop}.

\begin{table}[htbp]
\caption{p-values of testing the change in covariance structure before and after the Bitcoin important events - 1 hour data}
    \label{pvals}
    \begin{adjustbox}{width=\linewidth,center}
    \begin{tabular}{@{}llllp{3in}lll@{}}
\toprule
Period I start date  & Period II start date& Date of event ($T_0$) & Period II end date & Event description & $p_{[T_0-30D, T_0]}$ & $p_{[T_0-15D, T_0+15D]}$ \\ \midrule
9 October 2017 & 24 October 2017 & 8 November 2017 & 23 November 2017 & Developers cancel splitting of Bitcoin. & 0.31 & 0.05 \\
28 November 2017 & 13 December 2017 & 28 December 2017 & 12 January 2018 & South Korea announces strong measures to regulate trading of cryptocurrencies. & 0.32 & 0.27 \\
14 December 2017 & 28 December 2017 & 13 January 2018 & 28 January 2018 & Announcement that 80\% of Bitcoin has been mined. & 0.21 & 0.22 \\
31 December 2017 & 15 January 2018 & 30 January 2018 & 14 February 2018 & Facebook bans advertisements promoting, cryptocurrencies. & 0.17 & 0.82 \\
5 February 2018 & 20 February 2018 & 7 March 2018 & 22 March 2018 & The US Securities and Exchange Commission says it is necessary for crypto trading platforms to register. & 0.03 & 0.01 \\
12 February 2018 & 28 February 2018 & 14 March 2018 & 29 March 2018 & Google bans advertisements promoting cryptocurrencies. & 0.50 & 0.34 \\ \bottomrule
\end{tabular}
    \end{adjustbox}
\end{table}

\begin{table}[htbp]
\caption{p-values of testing the change in covariance structure before and after the Bitcoin important events - 1 minute data}
\label{pvals1min}
\centering
\begin{adjustbox}{width=\linewidth,center}
\begin{tabular}{@{}lp{2.5in}llll@{}}
\toprule
Date of event ($T_0$) & Event description & $p_{[T_0-2D, T_0-1D]}$ & $p_{[T_0-1D, T_0]}$ & $p_{[T_0,  T_0+1D]}$ \\ \midrule
8 November 2017 & Developers cancel splitting of Bitcoin. & 0.0523 & 0.0832 & 0.2049 \\
28 December 2017 & South Korea announces strong measures to regulate the trading of cryptocurrencies. & 0.5889 & 0.0027 & 0.0080 \\
13 January 2018 & Announcement that 80\% of Bitcoin has been mined. & 0.0300 & 0.0035 & 0.0493 \\
30 January 2018 & Facebook bans advertisements promoting cryptocurrencies. & 0.8224 & 0.0352 & 0.7774 \\
7 March 2018 & The US Securities and Exchange Commission says it is necessary for crypto trading platforms to register. & 0.0029 & 0.0225 & 0.6398 \\
14 March 2018 & Google bans advertisements promoting cryptocurrencies. & 0.403 & 0.3019 & 0.0453 \\ \bottomrule
\end{tabular}
\end{adjustbox}
\end{table}
\begin{figure}[htbp]
    \centering
    \subfigure{\includegraphics[width=0.49\textwidth]{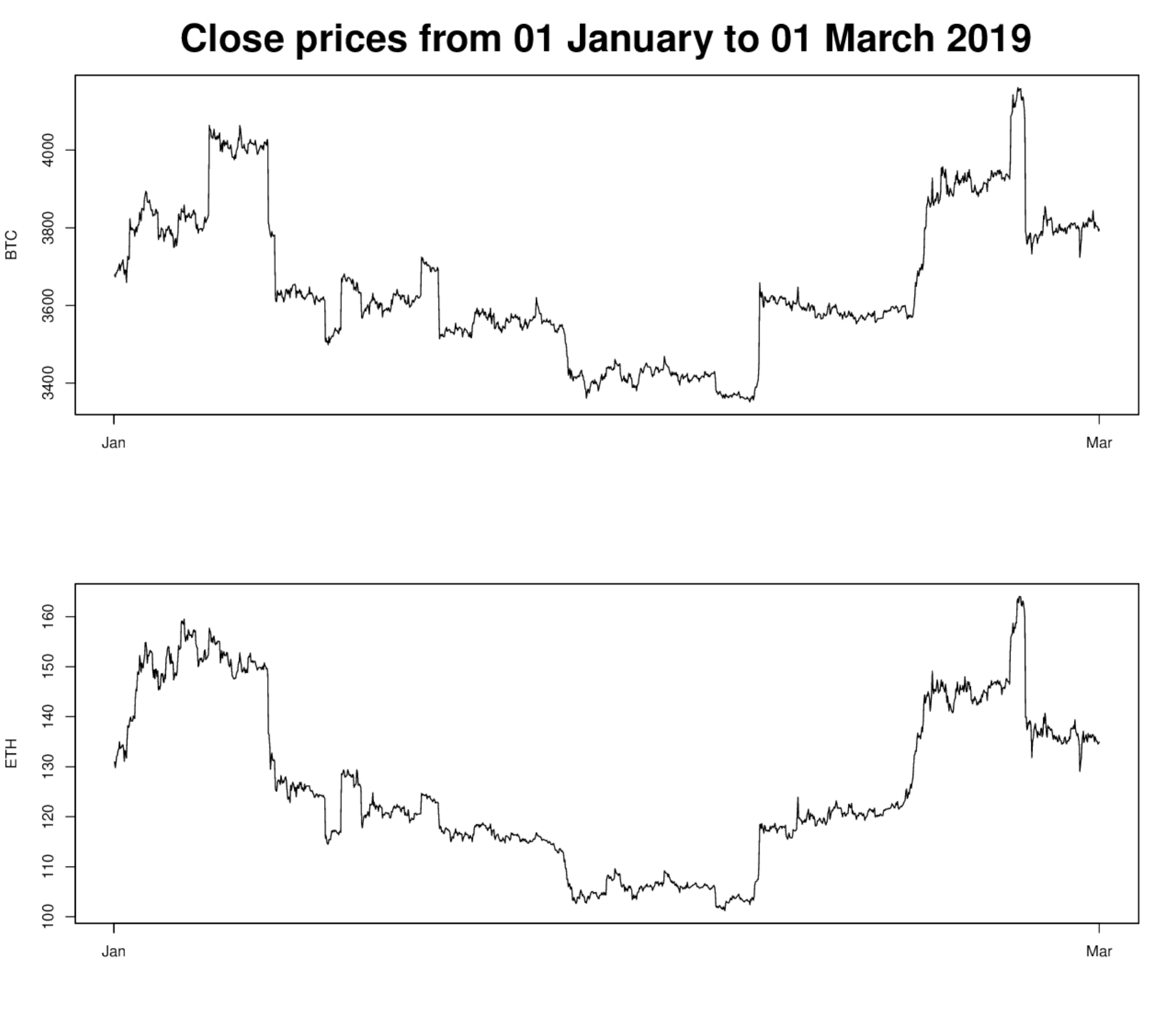}} 
    \subfigure{\includegraphics[width=0.49\textwidth]{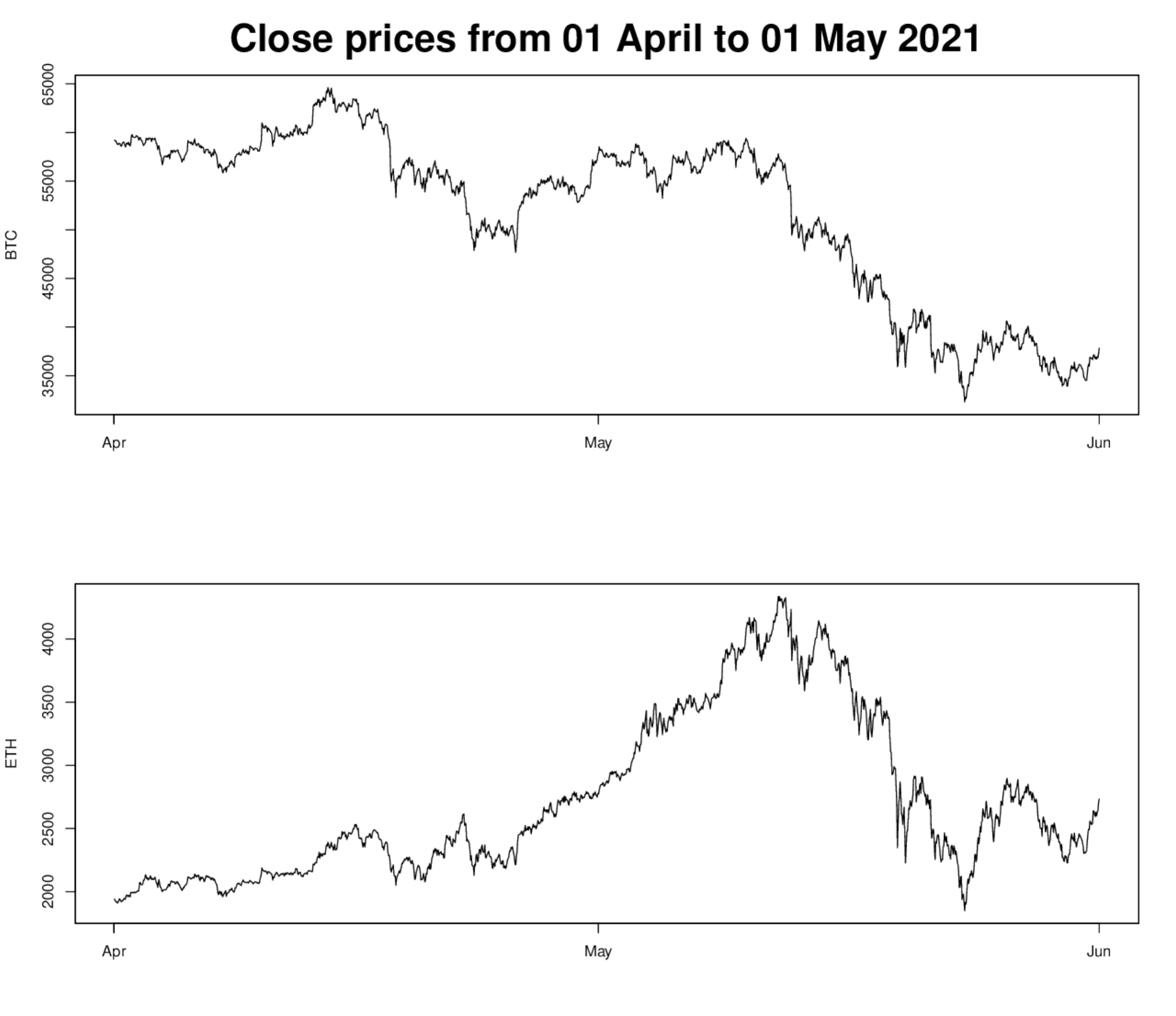}} 
    \caption{(a) Period 1 (b) Period 2}
    \label{fig:foobar}
\end{figure}
\begin{figure}[htbp]
    \centering
    \subfigure{\includegraphics[width=0.49\textwidth]{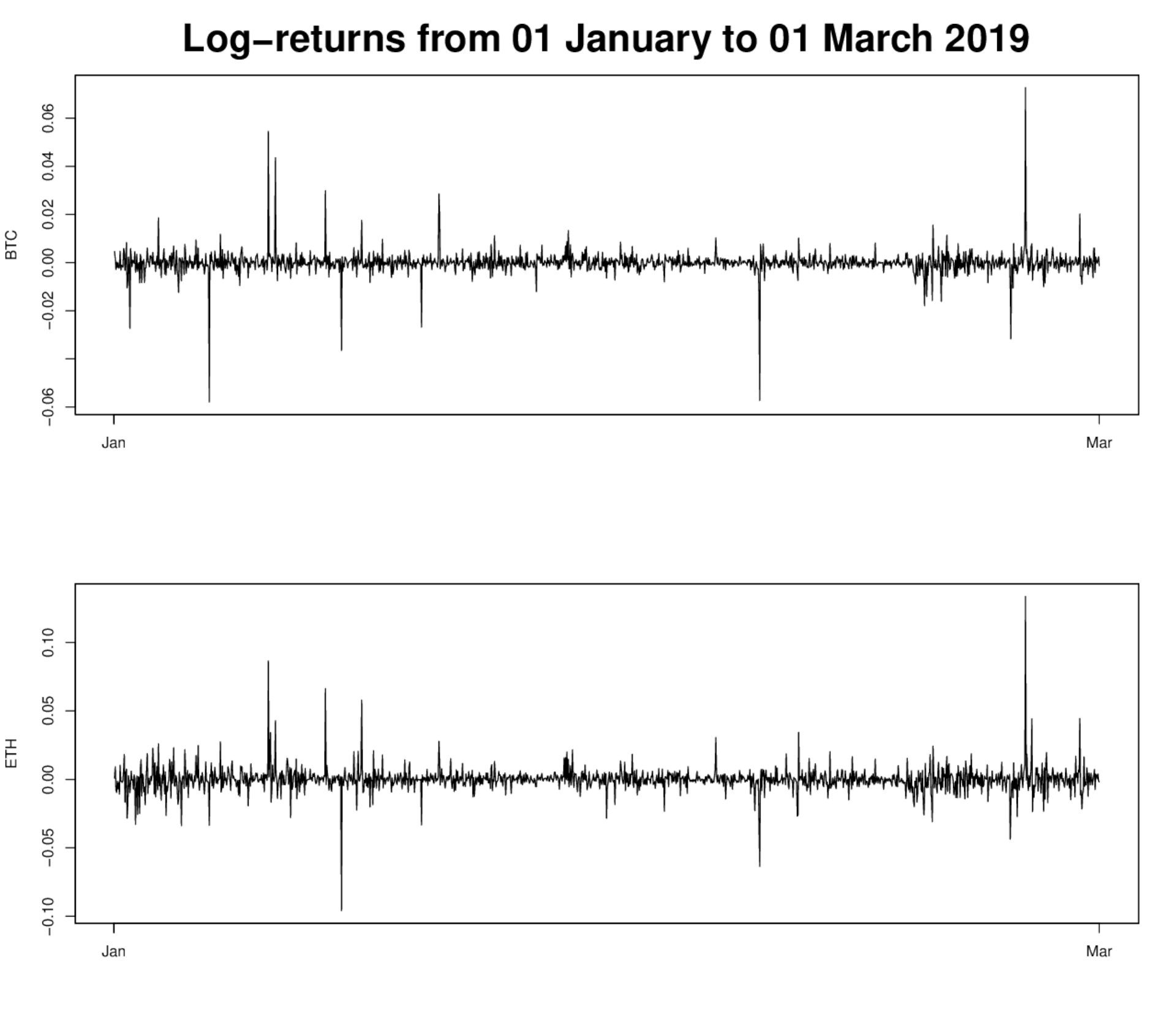}} 
    \subfigure{\includegraphics[width=0.49\textwidth]{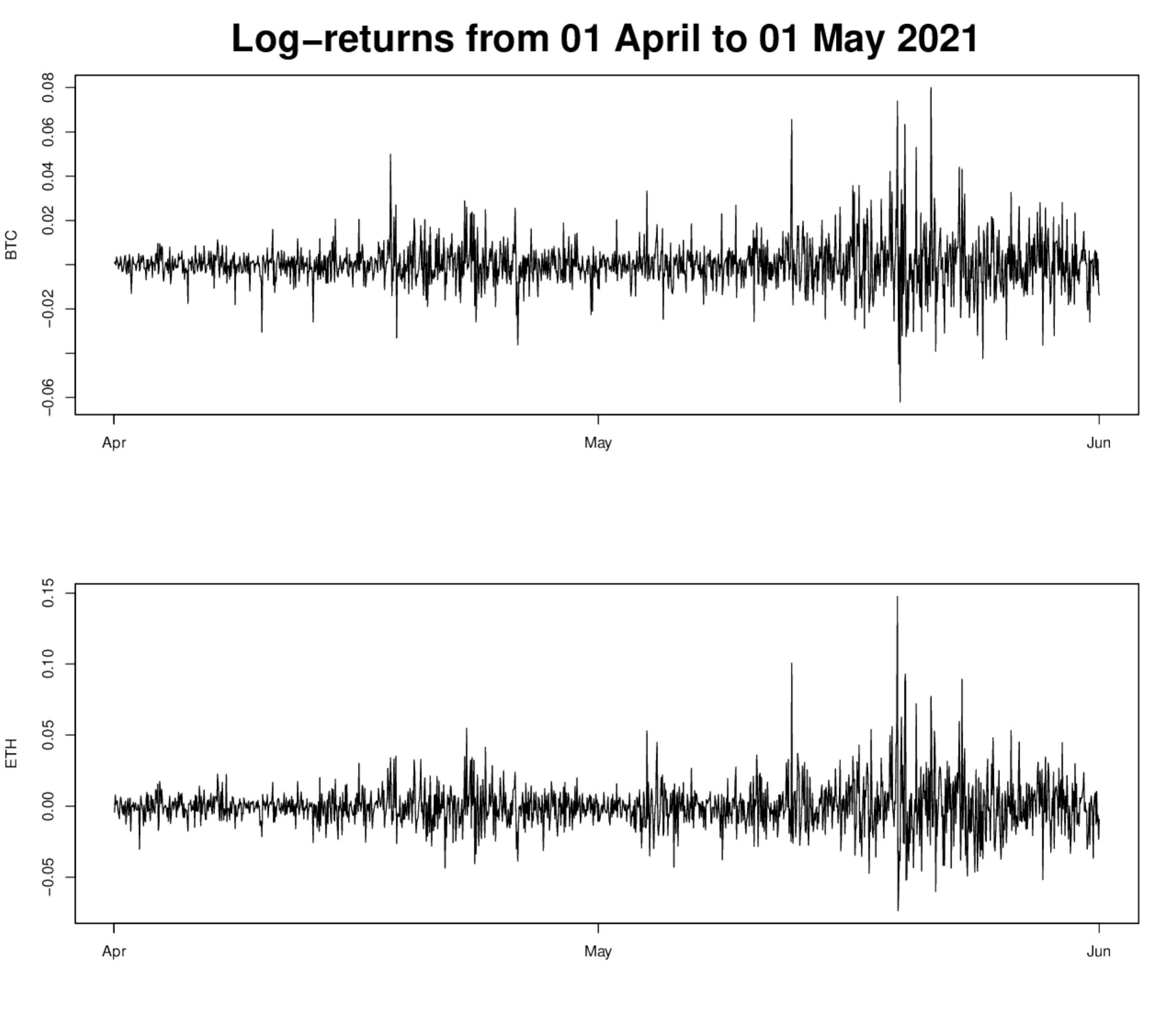}} 
    \caption{(a) Period 1 (b) Period 2}
    \label{fig:fooLR}
\end{figure}
{ 
The second data set consists of the stock data of the three biggest S\&P 500 companies at the moment. We calculate the daily log-returns for the closing prices of Apple Inc. (AAPL), Microsoft (MSFT), and Amazon (AMZN) from January 1, 2021, to January 1, 2023, covering 503 trading days. The data are sourced from Yahoo Finance (\url{https://finance.yahoo.com}). We partition the data into blocks of 7 trading days and compute a covariance matrix for each block. We then test if there is a significant change in the covariance structure between the first 31 blocks (January 1, 2021, to November 19, 2021) and the remaining 30 blocks (November 20, 2021, to January 1, 2023), obtaining a p-value of 0.001. The outcome suggests that there is a significant change in the covariance structure of the three largest companies in the S\&P 500 \cite{peli2023assessing}.
}

\section*{Conclusion}
In this paper, we propose the first test statistic for testing the equality of matrix distributions in the space of symmetric positive definite matrices. Our test provides important insights into the uncertainty of covariance estimates and is able to differentiate between different distributions with the same expectation. While the power study in higher dimensions is left for further research due to computational demands, our real data examples demonstrate the test's ability to detect changes in the time series of logarithmic returns, which could prove to be significant in change-point problems.


\section*{Declaration of interest}
The authors have no interest to declare.
\section*{Acknowledgements}
The authors would like to express their deepest gratitude to Professor Donald  Richards for his comments, which have significantly improved the paper.

The work of B. Milo\v sevi\'c is  supported by the Ministry of Science, Technological Development and Innovations of the Republic of Serbia (the contract 451-03-47/2023-01/ 200104), and  by the COST action
CA21163 - Text, functional and other high-dimensional data in econometrics: New models, methods, applications (HiTEc).

\end{document}